\documentclass[11pt]{llncs}

\newtheorem{fact}{Fact}
\newtheorem{observation}[theorem]{Observation}

\usepackage{epsfig}
\usepackage{amsmath}
\usepackage{amsfonts}

\def\nottoobig#1{{\hbox{$\left#1\vcenter to1.111\ht\strutbox{}\right.\n@space$}}}

\if01

\newtheorem{theorem}{Theorem}[section]

\newtheorem{lemma}[theorem]{Lemma}
\newtheorem{claim}[theorem]{Claim}
\newtheorem{proposition}[theorem]{Proposition}

\newtheorem{definition}[theorem]{Definition}

\fi

\newcommand{\nat}{{\mathbb N}}

\newcommand{\ie}{$\mbox{i.e.}$}

\newlength{\filength}
\newsavebox{\gcbox}
\sbox{\gcbox}{\framebox[\filength]{\rule{0ex}{2ex}}}

\newcommand{\qedblob}{\mbox{\rule[-1.5pt]{5pt}{10.5pt}}}
\def\literalqed{{\ \nolinebreak\hfill\mbox{\qedblob\quad}}}

\def\qed{\literalqed}

\newcommand{\singlespacing}{\let\CS=
\@currsize\renewcommand{\baselinestretch}{1}\tiny\CS}
\newcommand{\singlespacingplus}{\let\CS=
\@currsize\renewcommand{\baselinestretch}{1.25}\tiny\CS}
\newcommand{\doublespacing}{\let\CS=
\@currsize\renewcommand{\baselinestretch}{1.75}\tiny\CS}
\newcommand{\draftspacing}{\let\CS=
\@currsize\renewcommand{\baselinestretch}{2.0}\tiny\CS}

\def\zo{\{0,1\}}

\def\mapping{\rightarrow}

\newcommand{\zon}{\zo^n}

\newcommand{\polylog}{{\rm polylog}}
\newcommand{\poly}{{\rm poly}}

\makeatletter
\def\@listI{\leftmargin\leftmargini \parsep 4.5pt plus 1pt minus 1pt\topsep6pt plus 2pt minus 2pt \itemsep  2pt plus 2pt minus 1pt}

\let\@listi\@listI
\@listi
\makeatother

\author{ {Marius Zimand\/}
\thanks{  Department of Computer and Information Sciences, Towson University,
Baltimore, MD.; email: mzimand@towson.edu; http://triton.towson.edu/\~{ }mzimand.
Part of this work has been supported
by NSF grant CCF 0634830.}}

\author{
{Marius Zimand}
\thanks{ {  \tt  http://triton.towson.edu/\~{ }mzimand}.}}
\institute{
{Department of Computer and Information Sciences, Towson University,
Baltimore, MD, USA}
}
\begin{document}
\title{Counting dependent and independent strings}

\date{}

\maketitle

\pagestyle{plain}
\pagenumbering{arabic}

\begin{abstract} 
We derive quantitative results regarding sets of $n$-bit strings that have different dependency or independency properties. Let $C(x)$ be the Kolmogorov complexity of the string $x$. A string $y$ has $\alpha$ dependency with a string $x$ if $C(y) - C(y \mid x) \geq \alpha$. A set of strings $\{x_1, \ldots, x_t\}$ is pairwise $\alpha$-independent if for all $i\not=j$, $C(x_i) - C(x_i \mid x_j) \leq \alpha$. A tuple of strings $(x_1, \ldots, x_t)$ is mutually $\alpha$-independent if $C(x_{\pi(1)} \ldots x_{\pi(t)}) \geq C(x_1) + \ldots + C(x_t) - \alpha$, for every permutation $\pi$ of $[t]$.  We show that:
\begin{itemize}
	\item For every $n$-bit string $x$ with complexity $C(x) \geq \alpha + 7 \log n$, the set of $n$-bit strings that have $\alpha$ dependency with $x$ has size at least $(1/\poly(n)) 2^{n-\alpha}$. In case $\alpha$ is computable from $n$ and $C(x) \geq \alpha + 12 \log n$, the size of same set is at least $(1/C)2^{n-\alpha} - \poly(n) 2^{\alpha}$, for some positive constant $C$.
	\item There exists a set of $n$-bit strings $A$ of size $\poly(n)2^{\alpha}$ such that any $n$-bit string has $\alpha$-dependency with some string in $A$.
	\item If the set of $n$-bit strings $\{x_1, \ldots, x_t\}$ is  pairwise $\alpha$-independent, then $t \leq \poly(n) 2^{\alpha}$. This bound is tight within a $\poly(n)$ factor, because, for every $n$, there exists a set of $n$-bit strings $\{x_1, \ldots, x_t\}$ that is pairwise $\alpha$-dependent with $t = (1/\poly(n))\cdot 2^{\alpha}$ (for all $\alpha \geq 5 \log n$).
	\item If the tuple of $n$-bit strings $(x_1, \ldots, x_t)$ is mutually $\alpha$-independent, then $t \leq \poly(n) 2^{\alpha}$ (for all $\alpha \geq 7 \log n + 6$).
	\end{itemize}
\end{abstract}

\section{Introduction}
A fact common to many mathematical settings is that in a sufficiently large set some relationship emerges among its elements. Generically, these are called Ramsey-type results. We list just a few examples: any $n+1$ vectors in an $n$-dimensional vector space must be dependent; for every $k$ and sufficiently large $n$, any subset of $[n]$ of constant density must have $k$ elements in arithmetic progression; any set of $5$ points in the plane must contain $4$ points that form a convex polygon. All these results show that in a sufficiently large set, some attribute of one element is determined by the other elements. 

We present in this paper a manifestation of this phenomenon in the very general framework of algorithmic information theory. We show that in a sufficiently large set some form of algorithmical dependency among its elements must exist. Informally speaking, $\poly(n)\cdot 2^{\alpha}$ binary strings of length $n$ must share at least $\alpha$ bits of information. For one interpretation of ``share", we also show that this bound is tight within a $\poly(n)$ factor.

Central to our investigation are the notions of information in a string and the derived notion of dependency between strings. The information in a string $x$ is captured by its Kolmogorov complexity $C(x)$.
 A string $y$ has $\alpha$-dependency with string $x$ if
$C(y) - C(y \mid x) \geq \alpha$. The expression $C(y) - C(y \mid x)$, denoted usually more concisely as $I(x : y)$, represents \emph{the quantity of information in $x$ about $y$} and is a key concept in information theory. It is known that $I(x:y) = I(y:x) \pm O(\log n)$ (Symmetry of Information Theorem~\cite{zvo-lev:j:kol}), where $n$ is the length of the longer between the strings $x$ and $y$, and therefore $I(x:y)$ is also called the mutual information of $x$ and $y$. For any $n$-bit string $x$ and positive integer $\alpha$, 
we are interested in estimating the size of the set $A_{x, \alpha}$ of $n$-bit strings $y$ such that $C(y) - C(y \mid x) \geq \alpha$.  One can see by a standard counting argument that $|A_{x, \alpha}| \leq 2^{n-\alpha + c}$ for some constant $c$. Regarding a lower bound for $|A_{x,\alpha}|$, it is easy to see that if $C(x) \preceq \alpha$, then $A_{x, \alpha}$ is empty (intuitively, in order for $x$ to have $\alpha$ bits of information about $y$, it needs to have $\alpha$ bits of information to start with, regardless of $y$). The lower bound that we establish holds for any string having Kolmogorov complexity $\succeq \alpha$.\footnote{We use notation $\poly(n)$ for $n^{O(1)}$ and $\approx$, $\preceq$ and $\succeq$ to denote that the respective equality or inequality holds with an error of at most $O(\log n)$.}  For such strings $x$, we show that $|A_{x, \alpha}| \geq (1/\poly(n)) 2^{n-\alpha}$. A related set is $B_{x, \alpha}$ consisting of the $n$-bit strings $y$ with the property $C(y \mid n) - C(y \mid x) \geq \alpha$. This is the set of $n$-bit strings about which $x$ has $\alpha$ bits of information besides the length. Note that $B_{x, \alpha} \subseteq A_{x, \alpha}$. The same observations regarding an upper bound for $|B_{x, \alpha}|$ and the emptiness of $B_{x, \alpha}$ in case $C(x) \preceq \alpha$ remain valid. For $x$ with $C(x) \succeq \alpha$ and $\alpha$ computable from $n$, we show the lower bound
$|B_{x, \alpha}| \geq (1/C)\cdot 2^{n-\alpha} - \poly(n) \cdot 2^{\alpha}$, for some positive constant $C$.

 We turn  to the Ramsey-type results announced above. A set of $n$-bit strings $\{x_1, \ldots, x_t\}$ is pairwise $\alpha$-independent if for all $i \not=j$, $C(x_i) - C(x_i \mid x_j) \leq \alpha$. Intuitively, this means that any two strings in the set have in common at most $\alpha$ bits of information. For the notion of mutual independence we propose the following definition (but other variants are conceivable). The tuple of $n$-bit strings $(x_1, \ldots, x_t) \in (\zo^n)^t$ is mutually $\alpha$-independent if $C(x_{\pi(1)} \ldots x_{\pi(t)}) \geq C(x_1) + \ldots + C(x_t) - \alpha$, for every permutation $\pi$ of $[t]$.  Intuitively this means that $x_1, \ldots, x_t$ share at most $\alpha$ bits of information. We show that if $\{x_1, \ldots, x_t\}$ is pairwise $\alpha$-independent or if $(x_1, \ldots, x_t)$ is mutually $\alpha$-independent then $t \leq \poly(n) 2^{\alpha}$. The bound in the pairwise independent case is tight within a polynomial factor.

We also show that there exists a set $B$ of size $\poly(n)2^{\alpha}$ that ``$\alpha$-covers" the entire set of $n$-bit strings, in the sense that for each $n$-bit string $y$ there exists a string $x$ in $B$ that has $\alpha$ bits of information about $y$ (\ie, $y$ is in $A_{x, \alpha}$).

The main technical novelty of this paper is the technique used to lower bound the size of $B_{x, \alpha} = \{y \in \zo^n \mid C(y \mid n) - C(y \mid x) \geq \alpha\}$, which should be contrasted with a known and simple approach. This ``normal'' and simple approach is best illustrated when $x$ is random. In this case, the prefix $x(1:\alpha)$ of $x$ of length $\alpha$ is also random and, therefore, if we take $z$ to be an $(n-\alpha)$ long string that is random conditioned by $x(1:\alpha)$, then $C(z x(1:\alpha)) = n - O(\log n)$, $C(zx(1:\alpha) \mid x(1:\alpha)) = n - \alpha - O(\log n)$, and thus, $zx(1:\alpha) \in B_{x, \alpha + O(\log n)}$. There are approximately $2^{n - \alpha}$ strings $z$ as above, and this leads to a lower bound of $2^{n-\alpha}$ for $|B_{x, \alpha + O(\log n)}|$, which implies a lower bound of $(1/\poly(n))2^{n-\alpha}$ for $|B_{x,\alpha}|$. This method is so basic and natural that it looks hard to beat. However, using properties of Kolmogorov complexity extractors, we derive a better lower bound for $|B_{x,\alpha}|$ that does not have the slack of $1/\poly(n)$, in case $\alpha$ is computable from $n$ (even if $\alpha$ is not computable from $n$, the new method gives a tighter estimation than the above ``normal'' method).  
A Kolmogorov complexity extractor is a function that starting with several strings that have Kolmogorov complexity relatively small compared to their lengths, computes a string that has Kolmogorov complexity almost close to its length. A related notion, namely multi-source randomness extractors, has been studied extensively in computational complexity (see\cite{bou:j:multiextract,bar-imp-wig:c:multisourceext,bkssw:c:multisourceextract,raz:c:multiextract,rao:c:multisourceextract}). Hitchcock, Pavan and Vinodchandran~\cite{hit-pav-vin:t:Kolmextraction} have shown that Kolmogorov complexity extractors are equivalent to a type of functions that are close to being multisource randomness extractors. Fortnow, Hitchcock, Pavan, Vinodchandran and Wang~\cite{fhpvw:c:extractKol} have constructed a polynomial-time Kolmogorov complexity extractor based on the multi-source randomness constractor of Barak, Impagliazzo and Wigderson~\cite{bar-imp-wig:c:multisourceext}. 
The author has constructed Kolmogorov complexity extractors for other settings, such as extracting from infinite binary sequences~\cite{zim:c:csr,zim:c:kolmlimindep} or from binary strings that have a bounded degree of dependence~\cite{zim:c:kolmlimindep,zim:c:genindepstringsCiE09,zim:c:impossibamplific}. The latter type of Kolmogorov complexity extractors is relevant for this paper. Here we modify slightly an extractor $E$ from~\cite{zim:c:impossibamplific}, which, on inputs two $n$-bit strings $x$ and $y$ that have Kolmogorov complexity at least $s$ and dependency at most $\alpha$, constructs an $m$-bit string $z$ with $m \approx s$ and Kolmogorov complexity equal to $m - \alpha - O(1)$ even conditioned by any one of the input strings. Let us 
call a pair of strings $x$ and $y$ with the above properties as \emph{good-for-extraction}. We fix $x \in \zo^n$ with $C(x) \geq s$. Let $z$ be the most popular image of the function $E$ restricted to $\{x\} \times \zo^n$.  Because it is distinguishable from all other strings, given $x$, $z$ can be described with only $O(1)$ bits (we only need a description of the function $E$ and of the input length). Choosing $m$ just slightly larger than $\alpha$ we arrange that $C(z \mid x) < m-\alpha -O(1)$ . This implies that all the preimages of $z$ under  $E$ restricted as above are
\emph{bad-for-extraction}.  Since the size of $E^{-1}(z) \cap (\{x\}\ \times \zo^n)$ is at least $2^{n-m}$, we see that at least $2^{n-m}$ pairs $(x,y)$ are bad-for-extraction. A pair of strings $(x,y)$ is bad-for-extraction if either $y$ has Kolmogorov complexity below $s$ (and it is easy to find an upper bound on the number of such strings), or if $y \in B_{x, \alpha}$. This allows us to find the lower bound  for the size of $B_{x,\alpha}$.
\section{Preliminaries}

We work over the binary alphabet $\{0,1\}$; $\nat$ is the set of natural numbers. A string $x$ is an element of $\{0,1\}^*$; $|x|$ denotes its length; $\zo^n$ denotes the set of strings of length $n$; $|A|$ denotes the cardinality of a finite set $A$; for $n \in \nat$, $[n]$ denotes the set $\{1,2, \ldots, n\}$. We recall the basics of (plain) Kolmogorov complexity (for an extensive coverage, the reader should consult one of the monographs by Calude~\cite{cal:b:infandrand}, Li and Vit\'{a}nyi~\cite{li-vit:b:kolmbook}, or Downey and Hirschfeldt~\cite{dow-hir:b:algrandom}; for a good and concise introduction, see Shen's lecture notes~\cite{she:t:kolmnotes}). Let $M$ be a standard Turing machine. For any string $x$, define the \emph{(plain) Kolmogorov complexity} of $x$ with respect to $M$, as 
\[C_M(x) = \min \{ |p| \mid M(p) = x \}.
\]
 There is a universal Turing machine $U$ such that for every machine $M$ there is a constant $c$ such that for all $x$,
\begin{equation}
\label{e:univ}
C_U(x) \leq C_M(x) + c.
\end{equation}
We fix such a universal machine $U$ and dropping the subscript, we let $C(x)$ denote the Kolmogorov complexity of $x$ with respect to $U$. We also use the  concept of conditional Kolmogorov complexity. Here the underlying machine is a Turing machine that in addition to the read/work tape which in the initial state contains the input $p$, has a second tape containing initially a string $y$, which is called the conditioning information. Given such a machine $M$, we define the Kolmogorov complexity of $x$ conditioned by $y$ with respect to $M$ as 
\[C_M(x \mid y) = \min \{ |p| \mid M(p, y) = x \}.
\]
Similarly to the above, there exist  universal machines of this type and they satisfy the relation similar to Equation~\ref{e:univ}, but for conditional complexity. We fix such a universal machine $U$, and dropping the subscript $U$, we let $C(x \mid y)$ denote the Kolmogorov complexity of $x$ conditioned by $y$ with respect to $U$. 

There exists a constant $c_U$ such that for all strings $x$, $C(x) \leq |x| + c_U$. Strings $x_1, x_2, \ldots, x_k$ can be encoded in a self-delimiting way (\ie, an encoding from which each string can be retrieved) using $|x_1| + |x_2| + \ldots + |x_k| + 2 \log |x_2| + \ldots + 2 \log |x_k| + O(k)$ bits. For example, $x_1$ and $x_2$ can be encoded as $\overline{(bin (|x_2|)} 01 x_1 x_2$, where $bin(n)$ is the binary encoding of the natural number $n$ and, for a string $u = u_1 \ldots u_m$, $\overline{u}$ is the string $u_1 u_1 \ldots u_m u_m$ (\ie, the string $u$ with its bits doubled).
\if01
For every sufficiently large $n$ and $k \leq n$,
\[
2^{k-2\log n} < |\{x \in \zo^n \mid C(x) \leq k\}| < 2^{k+1}. 
\]
\fi

Given a string $x$ and its Kolmogorov complexity $C(x)$, one can effectively enumerate all descriptions $y$ of $x$ of length $C(x)$, \ie, the set $\{y \in \zo^{C(x)} \mid U(y) = x\}$. We denote $x^*$ the first string in this enumeration. Note that
$C(x) - O(1) \leq C(x^*) \leq |x^*| +O(1)= C(x)  + O(1)$.

The Symmetry of Information Theorem~\cite{zvo-lev:j:kol} states that for any two strings $x$ and $y$, 
\begin{itemize}
\item[(a)] $C(xy) \leq C(y) + C(x \mid y) + 2 \log C(y) +O(1)$.
\item[(b)] $C(xy) \geq C(x) + C(y \mid x) - 2 \log C(xy) - 4 \log \log C(xy) - O(1)$.
\item[(c)] If $|x| = |y| = n$, $C(y) - C(y\mid x) \geq C(x) - C(x \mid y) - 5 \log n$
\end{itemize}
Since the theorem is usually stated in a slightly different form and since we use the constants specified above, we present in the appendix the proof (which follows the standard method).

As discussed in the Introduction, our main focus is on sets of strings having certain dependency or independency properties. For convenience, we restate here the main definitions.
\begin{definition}
The string $y$ has $\alpha$-dependency (where $\alpha \in \nat$) with the string $x$ if $C(y) - C(y \mid x) \geq \alpha$ or if $x$ coincides with $y$.
\end{definition}
We have included the case ``$x$ coincides with $y$" to make a string dependent with itself even in case it has low Kolmogorov complexity.
\begin{definition}
The strings $x_1, \ldots, x_t$ are pairwise $\alpha$-independent if for all $i \not= j$, $C(x_i) - C(x_i \mid x_j) \leq \alpha$.
\end{definition}
\begin{definition}
The tuple of strings $(x_1, \ldots, x_t)$ is mutually $\alpha$-independent (where $\alpha \in \nat$) if
$C(x_{\pi(1)} x_{\pi(2)} \ldots x_{\pi(t)}) \geq C(x_1) + C(x_2) + \ldots + C(x_t) - \alpha$, for every permutation $\pi$ of $[t]$.
\end{definition}
\section{Strings dependent with a given string}

Given a string $x \in \zo^n$, and $\alpha \in \nat$, how many strings have dependency with $x$ at least  $\alpha$? That is we are interested in estimating the size of the set
\[
A_{x, \alpha} = \{y \in \zo^n \mid C(y) - C(y \mid x) \geq \alpha \}.
\]
This is the set of strings about which, roughly speaking, $x$ has at least $\alpha$ bits of information. A related set is
\[
B_{x, \alpha} = \{y \in \zo^n \mid C(y \mid n) - C(y \mid x) \geq \alpha \},
\]
consisting of the $n$-bit strings about which $x$ provides $\alpha$ bits of information besides the length $n$. Clearly, $B_{x,\alpha} \subseteq A_{x, \alpha}$, and thus an upper bound for $|A_{x,\alpha}|$ also holds for $|B_{x,\alpha}|$, and a lower bound for $|B_{x,\alpha}|$ also holds for $|A_{x,\alpha}|$.

We show that for some polynomial $p$ and for some constant $C$, for all $x$ and $\alpha$ except some special values,
\[
(1/p(n)) \cdot 2^{n-\alpha} \leq |A_{x, \alpha}| \leq C 2^{n-\alpha}, 
\]
and, in case $\alpha(n)$ is computable from $n$,
\[
(1/C) \cdot 2^{n-\alpha} - p(n)2^{\alpha} \leq |B_{x, \alpha}| \leq C 2^{n-\alpha}, 
\]
The upper bounds for the sizes of  $A_{x, \alpha}$  and $B_{x, \alpha}$ can be readily derived. Observe that the set $A_{x, \alpha}$ is included in $\{ y \in \zo^n \mid C(y \mid x) < n - \alpha + c \}$ for some constant $c$, and therefore 
\[
|A_{x, \alpha}| \leq C \cdot 2^{n-\alpha},
\]
 for $C= 2^c$. 

We move to finding a lower bound for the size of $A_{x, \alpha}$.
A first observation is that for $A_{x, \alpha}$ to be non-empty, it is needed that $C(x) \succeq \alpha$. Indeed, it is immediate to observe that for any strings $x$ and $y$ of length $n$,
\[
C(y) \leq C(x) + C(y \mid x) + 2 \log C(x) + O(1) \leq C(x) + C(y \mid x) + 2 \log n + O(1),
\]
and thus, if $C(y) - C(y \mid x) \geq \alpha$, then $C(x) \geq \alpha - 2\log n - O(1)$. Intuitively, if the information in $x$ is close to $\alpha$, not too many strings can be $\alpha$-dependent with it.

We provide a lower bound for $|A_{x, \alpha}|$, for every string $x$ with $C(x) \geq \alpha + 7 \log n$. The proof uses the basic "normal" approach presented in the Introduction. To simplify the discussion, suppose $C(x) = \alpha$. Then if we take  a string $z$ of length $n-\alpha$ that is random conditioned by $x^*$, it holds that $C(x^*z) \approx n$ and $C(x^*z \mid x^*) \approx n- \alpha$. Thus, $C(x^*z) - C(x^* z \mid x^*) \succeq \alpha$.  Note that there are approximately $2^{n-\alpha}$ such strings $x^*z$. Since $x^*$ can be obtained from $x$ and $C(x)$, we can replace $x^*$ by $x$ in the conditioning at a small price. We obtain approximately $2^{n-\alpha}$ strings in $A_{x, \alpha}$.

\begin{theorem}
\label{t:firstestimate}
For every natural number $n$, for every natural number $\alpha$ and for every $x \in \zo^n$ such that $C(x) \geq \alpha + 7 \log n$,
\[
|A_{x, \alpha}| \geq \frac{1}{2n^7} 2^{n-\alpha},
\]
provided $n$ is large enough.
\end{theorem}
\begin{proof}
Let $k = C(x)$ and let $\beta = \alpha + 7 \log n$.  Let $x^*$ be the smallest description of $x$ as described in the Preliminaries.   Let $x^*_\beta$ be the prefix of $x^*$ of length $\beta$. Since $x^*$ is described by $x^*_\beta$ and by its suffix of length $k-\beta$, $C(x^*) \leq C(x^*_\beta) + (k-\beta) + 2 \log C(x^*_\beta) + O(1)$ and, thus
\[
\begin{array}{ll}
C(x^*_\beta) & \geq C(x^*) - (k - \beta) - 2 \log C(x^*_\beta) - O(1) \\
& \geq (k - O(1)) - (k -\beta) - 2 \log C(x^*_\beta) - O(1) \\
& \geq \beta - 2 \log \beta - O(1).
\end{array}
\]
The set $B = \{z \in \zo^{n - \beta} \mid C(z \mid x^*_\beta) \geq n - \beta  - 1\}$ has size at least $(1/2) \cdot 2^{n-\beta}$ (using a standard counting argument). Consider a string $y \in \zo^n$ of the from $y = x^*_\beta z$ with $z \in B$. There are at least $(1/2)\cdot 2^{n-\beta}$ such strings. 

By symmetry of information,
\[
\begin{array}{ll}
C(y) & = C(x^*_\beta z) \geq C(x^*_\beta) + C(z \mid x^*_\beta) - (2 \log n + 4 \log \log n +O(1)) \\
& \geq (\beta - 2 \log \beta) + (n- \beta -1) - (2 \log n + 4 \log \log n +O(1)) \\
& \geq n - (4 \log n + 4 \log \log n +O(1)) \geq n - 5 \log n.
\end{array}
\]
On the other hand, $C(y \mid x^*_\beta) = C(x^*_\beta z \mid x^*_\beta) \leq C(z) + O(1) \leq (n-\beta) + O(1)$. Note that 
\[
C(y \mid x) \leq C(y \mid x^*_\beta) + 2 \log n + 4 \log \log n + O(1),
\]
because one can effectively construct $x^*_\beta$ from $x, k$ and $\beta$. Therefore,
\[
C(y \mid x) \leq (n-\beta) + 2 \log n + 4 \log \log n +O(1),
\]
and thus
\[
C(y) - C(y \mid x) \geq \beta - (6 \log n + 8 \log \log n +O(1)) \geq \beta - 7 \log n.
\]
So, $y \in A_{x,\beta - 7 \log n} = A_{x, \alpha}$. Since this holds for all the strings $y$ mentioned above, it follows that
$|A_{x, \alpha}| \geq (1/2)2^{n-\beta} = (1/(2n^7)) \cdot 2^{n-\alpha}$.~\qed
\end{proof}
\smallskip

The lower bound for $|B_{x, \alpha}|$ is obtained using a technique based on Kolmogorov complexity extractors, as explained in the Introduction.
We use the following theorem which can be obtained by a simple modification of a  result from~\cite{zim:c:impossibamplific}.
\begin{theorem}
\label{t:kolmextractor} For any computable functions $s(n), m(n)$ and $\alpha(n)$ with $n \geq s(n) \geq \alpha(n) + 7 \log n $ and $m(n) \leq s(n) - 7 \log n$, there exists a computable ensemble of functions $E: \zo^n \times \zo^n \mapping \zo^{m(n)}$
such that for all $x$ and $y$ in $\zo^n$  
\begin{itemize}
\item if $C(x) \geq s(n), C(y \mid n) \geq s(n)$ and $C(y \mid n) - C(y \mid x) \leq \alpha(n)$

\item then $C(E(x,y) \mid x ) \geq m(n) - \alpha(n) - O(1)$.
\end{itemize}
\end{theorem} 
\begin{theorem}  
\label{t:conditionalestimate}
Let $\alpha(n)$ be a computable function. For every sufficiently large natural number $n$, for every $x \in \zo^n$ such that $C(x) \geq \alpha(n) + 8 \log n$, 
\[
|B_{x, \alpha(n)}| \geq \frac{1}{C}  \cdot 2^{n-\alpha(n)} - n^{8} 2^{\alpha(n)},
\]
for some positive constant $C$.
\end{theorem}
\begin{proof}
Let $m = \alpha(n) + c$ and $s = \alpha(n) + 8\log n$, where $c$ is a constant that will be specified later. 
Consider $E: \zo^n \times \zo^n \mapping \zo^m$ the Kolmogorov extractor given by Theorem~\ref{t:kolmextractor} for these parameters. Let $z \in \zo^m$ be the string that has the largest number of $E$ preimages in the set $\{x\} \times \zo^n$. Note that, for some constant $c_1$, $C(z \mid x) \leq c_1$, because, given $x$, $z$ can be constructed from a table of $E$, which at its turn can be constructed from $n$ which is given because it is the length of $x$. On the other hand, if $y \in \zo^n$ is a string with $C(y \mid n) \geq s$ and $C(y \mid n) - C(y \mid x) \leq \alpha(n)$, then Theorem~\ref{t:kolmextractor} guarantees that, for some constant $c_2$, $C(E(x,y) \mid x) \geq m - \alpha(n) - c_2 = c-c_2 > c_1$, for an appropriate $c$.
Therefore all the strings $y$ such that $E(x,y) = z$ are bad for extraction, \ie, they belong to 
\[
\{y \in \zo^n \mid C(y \mid n) < s\} \cup \{y \in \zo^n \mid C(y \mid n) \geq s \mbox{ and } C(y \mid n) - C(y \mid x) \geq \alpha\}.
\]
Since there are at least $2^{n-m}$ such strings $y$ and the first set above has less than $2^s$ elements, it follows that
\[
|\{y \in \zo^n \mid C(y\mid n) - C(y \mid x) \geq \alpha(n)\}| \geq 2^{n-m} - 2^{s} = \frac{1}{2^c}  \cdot 2^{n-\alpha(n)} - n^{8} 2^{\alpha(n)}.
\]
This concludes the proof.~\qed
\end{proof}
\if01
A similar proof using Kolmogorov extractors gives for many setting of parameters a better lower bound for $|A_{x, \alpha}|$ than the one from Theorem~\ref{t:firstestimate}.

\begin{theorem}  
\label{t:secondestimate}
Let $\alpha(n)$ be a computable function. For every sufficiently large natural number $n$, for every $x \in \zo^n$ such that $C(x) \geq \alpha(n) + 12 \log n$, 
\[
|A_{x, \alpha(n)}| \geq \frac{1}{2^{\tilde{C}(n)} (\alpha(n)+\tilde{C}(n))^8} \cdot 2^{n-\alpha(n)} - n^{12} 2^{\alpha(n)} \geq \frac{1}{n \cdot \polylog n}2^{n-\alpha(n)} - n^{12}2^{\alpha(n)}.
\]
\end{theorem}
\begin{proof} 
Let $m = \alpha(n) + \tilde{C}(n) + 4  \log (\alpha(n) + \tilde{C}(n))$ and $s = \alpha(n) + 12\log n$. 
Consider $E: \zo^n \times \zo^n \mapping \zo^m$ the Kolmogorov extractor given by Theorem~\ref{t:kolmextractor} for these parameters. Let $z \in \zo^m$ be the string that has the largest number of $E$ preimages in the set $\{x\} \times \zo^n$. Note that $C(z \mid x) \leq C(n) + O(1)$, because, given $x$, $z$ can be constructed from a table of $E$, which at its turn can be constructed from $n$. On the other hand, if $y \in \zo^n$ is a string with $C(y) \geq s$ and $C(y) - C(y \mid x) \leq \alpha(n)$, then Theorem~\ref{t:kolmextractor} guarantees that $C(E(x,y) \mid x) \geq m - \alpha(n) - 4 \log m > C(n) + O(1)$.
Therefore all the strings $y$ such that $E(x,y) = z$ are bad for extraction, \ie, they belong to $\{y \in \zo^n \mid C(y) < s\} \cup \{y \in \zo^n \mid C(y) \geq s \mbox{ and } C(y) - C(y \mid x) \geq \alpha\}$. Thus,
$|\{y \in \zo^n \mid C(y) - C(y \mid x) \geq \alpha(n)\}| \geq 2^{n-m} - 2^{s} = \frac{1}{2^{\tilde{C}(n)} (\alpha(n)+\tilde{C}(n))^8} \cdot 2^{n-\alpha(n)} - n^{12} 2^{\alpha(n)}$.~\qed

\end{proof}
\fi
\smallskip

The proof of Theorem~\ref{t:firstestimate} actually shows more: The lower bound applies even to a subset of $A_{x, \alpha}$ containing only strings with high Kolmogorov complexity. More precisely, if we denote $A_{x,\alpha, s} = \{y \in \zo^n \mid C(y) \geq s \mbox{ and } C(y) - C(y \mid x) \geq \alpha\}$, then $|A_{x, \alpha, n - 5 \log n}| \geq \frac{1}{2n^7}2^{n-\alpha}$.
Note that there is an interesting ``zone" for the parameter $s$ that is not covered by this result. Specifically, it would be interesting to lower bound the size of $A_{x,\alpha,n}$. This question remains open. Nevertheless, the technique from Theorem~\ref{t:conditionalestimate} can be used to tackle the variant in which access to the set $R = \{u \in \zo^n \mid C(u) \geq |u|\}$ is granted for free. Thus, let $A_{x,\alpha, n}^R = \{y \in \zo^n \mid C^R(y) \geq n \mbox{ and } C^R(y) - C^R(y \mid x) \geq \alpha\}$. 
\begin{proposition} For the same setting of parameters as in Theorem~\ref{t:conditionalestimate},
$|A_{x,\alpha, n}^R| \geq \frac{1}{C} \cdot 2^{n-\alpha(n)}$, for some positive constant $C$.
\end{proposition}
\begin{proof} Omitted from this extended abstract.~\qed
\if01
Sketch. We modify slightly the proof of Theorem~\ref{t:secondestimate}. First, Theorem~\ref{t:kolmextractor} can be modified to hold if the Kolmogorov complexity is relativized with $R$. Let $E$ be the Kolmogorov extractor obtained in this way. We take $z$ to be the string that has the largest number of $E$ preimages in the set $\{x\} \times \{y \in \zo^n \mid C^R(y) \geq n\}$, which is known to have size at least $2^{n-c}$ for some constant $c$. One can see that $C^R(z \mid x) \leq C(n) + O(1)$. The string $z$ has at least $2^{n-c-m}$ preimages as above, and their second component are all in $A_{x,\alpha, n}^R$, because all these preimages are bad for extraction.~\qed
\fi
\end{proof}

\section{Pairwise independent strings}
We show that if the $n$-bit strings $x_1, \ldots, x_t$ are pairwise $\alpha$-independent, then $t \leq \poly(n) 2^{\alpha}$. This upper bound is relatively tight, since there are sets with $(1/\poly(n)) \cdot 2^{\alpha}$ $n$-bit strings that are pairwise $\alpha$-independent.
\begin{theorem} 
\label{t:pairwiseubd}
For every sufficiently large $n$ and for every natural number $\alpha$, the following holds. If $x_1, \ldots, x_t$ are $n$-bit strings that are $\alpha$-independent, then $t < 2n^3 \cdot 2^\alpha$.
\end{theorem}
\begin{proof}
There are less than $2^{\alpha + 3 \log n}$ strings with Kolmogorov complexity less than $\alpha + 3 \log n$. We discard such strings from $x_1, \ldots, x_t$ and assume that $x_1, \ldots, x_{t'}$ are the strings that are left. Since $t < 2^{\alpha + 3 \log n} + t'$, we need to show that $t' \leq n^3 2^\alpha$.

For $1 \leq i \leq t'$, let $k_i = C(x_i)$ and let $x^*_i$ be the shortest description of $x_i$ as described in the Preliminaries.  Let $\beta = \alpha + 3 \log n$ (we assume that $\alpha \leq n - 3 \log n$, as otherwise the statement is trivial). We show that the prefixes of length $\beta$ of the strings  $x_1, \ldots, x_{t'}$ are all distinct, from which we conclude that $t' \leq 2^\beta = n^3 \cdot 2^\alpha$.

Suppose that there are two strings in the set that have equal prefixes of length $\beta$. W.l.o.g.\, we can assume that they are $x_1$ and $x_2$. Then
\[
C(x^*_1 \mid x^*_2) \leq (k_1 - \beta) + \log \beta + 2 \log \log \beta +O(1),
\]
because, given $x^*_2$, $x^*_1$ can be constructed from $\beta$ and the suffix of length $k_1 - \beta$ of $x^*_1$.
Note that 
\[
C(x^*_1 \mid x_2) \leq C(x^*_1 \mid x^*_2) + \log k_2 + 2 \log \log k_2 +O(1),
\]
because $x^*_2$ can be constructed from $x_2$ and $k_2$. Also note that $C(x_1 \mid x_2) \leq C(x^*_1 \mid x_2) + O(1)$.
Thus,
\[
C(x_1 \mid x_2) \leq C(x^*_1 \mid x^*_2) + \log k_2 + 2 \log \log k_2 + O(1).
\]
Therefore,
\[
\begin{array}{ll}
C(x_1) - C(x_1 \mid x_2) & \geq k_1 - (C(x^*_1 \mid x^*_2) + \log k_2 + 2 \log \log k_2 + O(1)) \\
& \geq k_1 - (k_1 - \beta) - \log \beta - 2 \log \log \beta - \log k_2 - 2 \log \log k_2 - O(1) \\
& \geq \beta - 3 \log n = \alpha,

\end{array}
\]
which is a contradiction.~\qed
\end{proof}

The next result shows that the upper bound in Theorem~\ref{t:pairwiseubd} is relatively tight. It relies on the well-known Tur\'{a}n's Theorem in Graph Theory~\cite{tur:j:indepset}, in the form due to Caro (unpublished) and Wei~\cite{wei:t:turan} (see~\cite[page 248]{juk:b:extremcombinat}): Let $G$ be a graph with $n$ vertices and let $d_i$ be the degree of the $i$-th vertex. Then $G$ contains an independent set of size at least $\sum \frac{1}{d_i+1}$.

\begin{theorem}
For every natural number $n$ and for every natural number $\alpha$ satisfying $5 \log n \leq \alpha \leq n$, there exists a constant $C$ and $t=\frac{1}{Cn^5}\cdot 2^\alpha$ $n$-bit strings $x_1, \ldots, x_t$ that are pairwise $\alpha$-independent.
\end{theorem}
\begin{proof}
Let $\beta = \alpha - 5 \log n$. Consider the graph $G=(V,E)$, where $V = \zo^n$ and $(u,v) \in E$ iff $C(u) - C(u \mid v) \geq \beta$ and $C(v) - C(v \mid u) \geq \beta$. Note that for every $u \in \zo^n$, the degree of $u$ is bounded by $|A_{u, \beta}| \leq 2^{n-\beta + c}$, for some constant $c$. Therefore, by Tur\'{a}n's theorem, the graph $G$ contains an independent set $I$ of size at least $2^n \cdot \frac{1}{2^{n-\beta + c} + 1} \geq 2^{\beta - c-1} = \frac{1}{Cn^5}\cdot 2^{\alpha}$.
For any two elements $u, v$ in $I$, we have either $C(u) - C(u\mid v) < \beta$ or $C(v) - C(v \mid u) < \beta$. In the second case, by symmetry of information, $C(u) - C(u \mid v) < \beta + 5 \log n = \alpha$. It follows that the strings in $I$ are pairwise $\alpha$-independent.~\qed
\end{proof}
\section{Mutually independent strings}
In this section we show that the size of a mutually $\alpha$-independent tuple of   $n$-bit strings is bounded by $\poly(n) 2^{\alpha}$.

For $u \in \zo^n$, we define $D_\alpha(u) = \{x \in \zo^n \mid u \in A_{x,\alpha}\} = \{x \in \zo^n \mid C(u) - C(u \mid x) \geq \alpha \}$
and $d_{\alpha}(u) = |D_{\alpha}(u)|$. 

\begin{lemma} 
\label{l:degree}
For every natural number  $n$ sufficiently large, for every natural number $\alpha$, and for every $u \in \zo^n$, with $C(u) \geq \alpha + 12 \log n$,
\[
\frac{1}{2n^{12}} 2^{n-\alpha} \leq d_\alpha (u) \leq n^5 \cdot 2^{n-\alpha}.
\]
\end{lemma}
\begin{proof}
For every $x \in A_{u, \alpha + 5 \log n}$,
\[
C(x) - C(x \mid u) \geq \alpha + 5 \log n
\]
which by symmetry of information implies
\[
C(u) - C(u \mid x) \geq \alpha + 5 \log n - 5 \log n = \alpha,
\]
and therefore, $u \in A_{x, \alpha}$. Thus
\[
d_\alpha(u) \geq |A_{u,\alpha + 5 \log n}| \geq \frac{1}{2n^7} 2^{n-\alpha - 5\log n} = \frac{1}{2n^{12}} 2^{n-\alpha}.
\]
For every $u \in \zo^n$,
\[
\begin{array}{ll}
x \in D_{u, \alpha} & \Rightarrow u \in A_{x, \alpha} \\
 & \Rightarrow C(u) - C(u \mid x) \geq \alpha \\

& \Rightarrow C(x) - C(x \mid u) \geq \alpha - 5 \log n \\

& \Rightarrow C(x  \mid u) \leq n - \alpha + 5 \log n.
\end{array}
\]
Thus, $d_\alpha(u) \leq |\{x \in \zo^n \mid C(x \mid u) \leq n - \alpha + 5 \log n\}| \leq n^5 \cdot 2^{n-\alpha}$.
~\qed \end{proof}

Since for any string $x$ and natural number $\alpha$, $|A_{x, \alpha}| \leq 2^{n - \alpha - c}$, for some constant $c$, it follows that we need
at least $T=2^{\alpha - c}$ strings $x_1, \ldots, x_T$ to ``$\alpha$-cover'' the set of $n$-bit strings, in the sense that for each $n$-bit string $y$, there exists $x_i$, $i \in [T]$ such that $y$ is $\alpha$-dependent with $x_i$. The next theorem shows that $\poly(n) 2^{\alpha}$ strings are enough to $\alpha$-cover the set of $n$-bit strings. 
\begin{theorem}
\label{t:sizecover}
For every natural number $n$ sufficiently large, for every natural number $\alpha$, there exists a set $B \subseteq \zo^n$ of size $\poly(n) 2^{\alpha}$  such that each string in 
$\zo^n$ is $\alpha$-dependent with some string in $B$, \ie, $\zo^n = \bigcup_{x \in B} A_{x, \alpha}$. More precisely the size of $B$ is bounded by $(2n^{13} + n^{12})\cdot 2^{\alpha}$.
\end{theorem}
\begin{proof}
(a) We choose $T= 2n^{13}2^\alpha$ strings $x_1, \ldots, x_T$, uniformly at random in $\zo^n$.
The probability that a fix $u$ with $C(u) \geq \alpha + 12 \log n$ does not belong to any of the sets $A_{x_i, \alpha}$, for $i \in [T]$, is at most
$(1 - \frac{1}{2n^{12} 2^\alpha})^T < e^{-n}$ (by Lemma~\ref{l:degree}). By the union bound, the probability that there exists
$u \in \zo^n$ with $C(u) \geq \alpha + 12 \log n$, that does not belong to any of the sets  $A_{x_i, \alpha}$, for $i \in [T]$, is bounded by $2^n \cdot e^{-n} < 1$. Therefore there are strings $x_1, \ldots, x_T$ in $\zo^n$ such that $\bigcup A_{x_i, \alpha}$ 
contains all the strings $u \in \zon$ having $C(u) \geq \alpha + 12 \log n$.  By adding to $x_1, \ldots, x_T$, the strings that have Kolmogorov complexity $< \alpha + 12 \log n$, we obtain the set $B$ that $\alpha$-covers the entire $\zo^n$.~\qed \end{proof}

\if01
PLAN TO ESTIMATE THE SIZE of a maximal set of independent strings: start with the set $B$ from above. Let $x$ be in $B$.
The number of independent strings in $A_{x, \alpha}$ should be around $2^{n/\alpha}$, because if $z_1, \ldots, z_T$ are such strings they each describe independently about $\alpha$ bits of $x$ (?). Actually we have $x$ is in the intersection of the sets $A_{z_i, \alpha}$, and such an intersection becomes empty if $T >> 2^{n/\alpha}$. 
\fi

To estimate the size of a mutually $\alpha$-independent tuple of strings, we need the following lemma. 

\begin{lemma}
\label{t:intersectestimate}
Let $\alpha, \beta \in \nat$ and let the tuple of $n$-bit strings $(x_1, x_2, \ldots, x_k)$  satisfy
$C(x_1 \ldots x_k) \geq C(x_1) + \ldots + C(x_k) - \beta$.
 Then there exists a constant $d$ such that
\[
| A_{x_1,\alpha} \cap \ldots \cap A_{x_k,\alpha}| \leq d n^{7k+5} k^3 2^{n - k \alpha + \beta}.
\]
\end{lemma}
\begin{proof}
Let $u \in \zo^n$ be a string in $A_{x_1,\alpha} \cap \ldots \cap A_{x_k,\alpha}$. Then
$C(u) - C(u \mid x_i) \geq \alpha$, for all $i \in [k]$. Therefore, by symmetry of information,
$C(x_i) - C(x_i \mid u) \geq \alpha - 5 \log n$, for all $i \in [k]$. It follows that for every $i \in [k]$, there exists a string $p_i$ of length $|p_i| \leq C(x_i) - \alpha + 5 \log n$ such that, given $u$, is a descriptor of $x_i$ (\ie, $U(p_i, u) = x_i$). The strings $p_1, \ldots, p_k$ describe the string $x_1 x_2 \ldots x_k$, given $u$, and therefore
\[
\begin{array}{ll}
C(x_1 x_2 \ldots x_k \mid u) &\leq |p_1| + \ldots + |p_k| + 2 \log |p_1| + \ldots + 2 \log |p_k| + O(1) \\
& \leq C(x_1) + \ldots + C(x_k) - k \alpha + 5k \log n + 2 \log |p_1| + \ldots + 2 \log |p_k| + O(1) \\
& \leq C(x_1) + \ldots + C(x_k) - k \alpha + 7k \log n + O(1) \\
&  \leq C(x_1 \ldots x_k) + \beta - k \alpha + 7 k \log n + O(1).
\end{array}
\]
So,
\[
C(x_1 \ldots x_k) - C(x_1 \ldots x_k \mid u) \geq -(\beta - k \alpha +7k \log n + O(1)).
\]

By symmetry of information, 
\[
\begin{array}{ll}
C(u) - C (u \mid x_1 \ldots x_k) & \geq C(x_1 \ldots x_k) - C(x_1 \ldots x_k \mid u) - 2 \log C(u) - 2 \log C(x_1 \ldots x_k u) \\
& \quad \quad  - 4 \log \log C(x_1 \ldots x_k u) - O(1).
\end{array}
\]

It follows that
\[
C(u) - C(u \mid x_1 \ldots x_k) \geq - (\beta  - k \alpha  + 7 k \log n ) - 5 \log n - 3 \log k
\]
and thus
\[
\begin{array}{ll} 
C(u \mid x_1 \ldots x_k) & \leq C(u) + \beta - k \alpha + (7k+5) \log n + 3 \log k \\
& \leq n + \beta - k \alpha + (7k+5) \log n + 3 \log k + O(1).
\end{array}
\]
Therefore,
\[
A_{x_1, \alpha} \cap \ldots \cap A_{x_k, \alpha} \subseteq \{u \in \zo^n \mid C(u \mid x_1 \ldots x_k) \leq n+\beta - k \alpha +(7k+5) \log n + 3 \log k +O(1)\}.
\]
The conclusion follows.
~\qed \end{proof}

Finally, we prove the upper bound for the size of a mutually $\alpha$-independent tuple of $n$-bit strings. 
\begin{theorem}
For every sufficiently large natural number $n$ the following holds. Let $\alpha$ be an integer such that $\alpha > 7 \log n + 6$. Let $(x_1, \ldots, x_t)$ be a mutually $\alpha$-independent tuple of $n$-bit strings. Then $t \leq \poly(n) 2^{\alpha}$.
\end{theorem}
\begin{proof}
By Theorem~\ref{t:sizecover},  there exists a set $B$ of size at most $\poly(n) 2^{\alpha + 5 \log n}$ such that every $n$-bit string $x$  is in  $A_{y, \alpha + 5 \log n}$, for some $y \in B$. We view $\{x_1, \ldots, x_t\}$ as a multiset. Let $y$ be the string  in $B$ that achieves the largest size of multiset $A_{y, \alpha + 5 \log n} \cap \{x_1, \ldots, x_t\}$ (we take every common element with the multiplicity in $\{x_1, \ldots, x_t\}$). Let $k$ be the size of the above intersection. Clearly, $k \geq t/|B|$. We will show that $k = \poly(n)$, and, therefore, 
$t \leq k \cdot |B| = \poly(n) \cdot 2^{\alpha}$.

 Without loss of generality suppose $A_{y, \alpha + 5 \log n} \cap \{x_1, \ldots, x_t\} = \{x_1, \ldots, x_k\}$ (as multisets).
Since, for every $i \in [k]$, $C(x_i) - C(x_i \mid y) \geq \alpha + 5 \log n$, by symmetry of information, it follows that $C(y) - C(y \mid x_i) \geq \alpha$. Thus
$y \in A_{x_1, \alpha} \cap \ldots \cap A_{x_k, \alpha}$. In particular, $A_{x_1, \alpha} \cap \ldots \cap A_{x_k, \alpha}$ is not empty. We want to use Lemma~\ref{t:intersectestimate} but before we need to estimate the difference between $C(x_1 \ldots x_k)$ and $C(x_1) + \ldots + C(x_k)$. 
\begin{claim}
$C (x_1 \ldots x_k) \geq C(x_1) + \ldots + C(x_k) - \beta$, where $\beta = \alpha + 4 \log (nt/2)$.
\end{claim}
\emph{Proof of claim.} Suppose $C(x_1 \ldots x_k) < C(x_1) + \ldots + C(x_k) - \beta$.
Note that
\[
\begin{array}{ll}
C(x_1 \ldots x_t) & \leq C(x_1 \ldots x_k) + C(x_{k+1} \ldots x_t) + 2 \log C(x_1 \ldots x_k) + O(1)\\
& \leq C(x_1) + \ldots C(x_k) + C(x_{k+1} \ldots x_t) - \beta + 2 \log kn + O(1).

\end{array}
\]
Since $C(x_1 \ldots x_t) \geq C(x_1) + \ldots + C(x_t) - \alpha$, it follows that
\[
C(x_{k+1}) + \ldots + C(x_t) - \alpha \leq C(x_{k+1} \ldots x_t)  - \beta + 2 \log kn + O(1).
\]
On the other hand,
\[
C(x_{k+1} \ldots x_t) \leq C(x_{k+1}) + \ldots + C(x_t) + 2 \log (t-k)n + O(1).
\]
It follows that
\[
\beta - \alpha \leq 2 \log kn + 2 \log (t-k)n + O(1).
\]
However, from the definition of $\beta$,
\[
\beta - \alpha = 4 \log (nt/2) > 2\log kn + 2\log (t-k)n + O(1).
\]
The contradiction proves the claim.~\qed
\smallskip

Now, by Lemma~\ref{t:intersectestimate},  

\[
\begin{array}{ll}
|A_{x_1, \alpha} \cap \ldots \cap A_{x_k, \alpha}| & \leq d n^{7k+5} k^3 2^{n - k \alpha + \beta} \\
& = d n^{7k+5} k^3 2^{n - (k-1)\alpha + 4 \log t + 4 \log (n/2) }\\
&\leq d n^{7k+5} k^3 2^{5n - (k-1)\alpha  + 4 \log (n/2)}, 
\end{array}
\]
where in the last line we used the fact that $t \leq 2^n$.

 It can be checked that if $\alpha > 7 \log n + 6$ and $k \geq n$, then the above upper bound is less than 1, which is a contradiction. It follows that $k < n$.~\qed \end{proof}
\section{Final remarks}
This paper provides tight bounds (within a polynomial factor) for the size of $A_{x, \alpha}$ (the set of $n$-bit strings that have $\alpha$-dependency with $x$) and for the size of sets of $n$-bit strings that are pairwise $\alpha$-independent. 

The size of a mutually $\alpha$-independent tuple of $n$-bit strings is at most $\poly(n) 2^{\alpha}$.  We do not know how tight this bound is and leave this issue as an interesting open problem.

We have recently learned about the paper~\cite{cha-lyu-ti-shen:j:kindep}, which obtains similar results regarding the size of sets of pairwise and $k$-independence strings, for a notion of independence that is suitable for strings with large Kolmogorov complexity.



\newpage

\appendix
\section{}
\medskip

{\bf Symmetry of Information Theorem}
\begin{theorem}
For any two strings $x$ and $y$, 
\begin{itemize}
\item[(a)] $C(xy) \leq C(y) + C(x \mid y) + 2 \log C(y) +O(1)$.
\item[(b)] $C(xy) \geq C(x) + C(y \mid x) - 2 \log C(xy) - 4 \log \log C(xy) - O(1)$.
\item[(c)] If $|x| = |y| = n$, $C(y) - C(y\mid x) \geq C(x) - C(x \mid y) - 5 \log n$
\end{itemize}
\end{theorem}

Proof (sketch): (a) is easy and (c) follows immediately from (a) and (b). We prove (b). Let $C(xy) = t$, $A = \{ (u,v) \mid C(uv) \leq t \}$, $A_u = \{ v \mid C(uv) \leq t \}$. Note that $|A| < 2^{t+1}$. Let
$e = \lfloor \log |A_x| \rfloor$. Let $B = \{u \mid |A_u| \geq 2^{e}\}$. Note that $x \in B$ and $|B| < |A|/2^e < 2^{t-e+1}$.

FACT: $x$ can be described by: $t$, rank in $B$ (which is written on exactly $t-e+1$ bits so that $e$ can be also reconstructed), $O(1)$ bits. So $C(x) \leq (t-e+1) + \log t + 2 \log \log t + O(1)$.

FACT: $y$, given $x$, can be described by: $t$, rank in $A_x$, $O(1)$ bits. So, $C(y \mid x) \leq e + \log t + 2 \log \log t +O(1)$.

Combining the last two: $C(x) \leq t - (C(y \mid x) -\log t - 2 \log \log t - O(1)) + \log t + 2 \log \log t + (1)$ 
$= C(xy) - C(y \mid x) + 2 \log t + 4 \log \log t +O(1)$
$=C(xy) - C(y \mid x) + 2 \log C(xy) + 4 \log \log C(xy) + O(1)$.

\end{document}